\documentclass[conference]{IEEEtran}

\usepackage{etoolbox}
\makeatletter
\patchcmd{\@makecaption}
  {\scshape}
  {}
  {}
  {}
\makeatother

\usepackage{hyperref}
\hypersetup
{
    colorlinks=true,
    linkcolor=black,
    filecolor=black,
    urlcolor=black,
    citecolor=black,
}

\IEEEoverridecommandlockouts
\usepackage{cite}
\usepackage{amsmath,amssymb,amsfonts}
\usepackage{algorithmic}
\usepackage{graphicx}
\usepackage{textcomp}
\usepackage{xcolor}
\usepackage{algorithm}
\usepackage{algorithmic}
\usepackage{pifont}
\usepackage{booktabs}
\usepackage{amsthm}
\usepackage{amsmath}
\newtheorem{theorem}{Theorem}
\newtheorem{lemma}{Lemma}
\usepackage{multirow}
\usepackage{caption}
\usepackage{url}
\usepackage[switch]{lineno}
\def\BibTeX{{\rm B\kern-.05em{\sc i\kern-.025em b}\kern-.08em
    T\kern-.1667em\lower.7ex\hbox{E}\kern-.125emX}}
\begin{document}

    \title{Seeing Is Not Always Believing: Invisible Collision Attack and  Defence on Pre-Trained Models}

    \makeatletter
    \author{
          \IEEEauthorblockN{Minghang Deng, Zhong Zhang, Junming Shao\IEEEauthorrefmark{1}\thanks{\IEEEauthorrefmark{1}Corresponding author}}
          \IEEEauthorblockA{University of University of Electronic Science and Technology of China}
    }

    \maketitle
    
    \begin{abstract}
        Large-scale pre-trained models (PTMs) such as BERT and GPT have achieved great success in diverse fields. The typical paradigm is to pre-train a big deep learning model on large-scale data sets, and then fine-tune the model on small task-specific data sets for downstream tasks. Although PTMs have rapidly progressed with wide real-world applications, they also pose significant risks of potential attacks. Existing backdoor attacks or data poisoning methods often build up the assumption that the attacker invades the computers of victims or accesses the target data, which is challenging in real-world scenarios. In this paper, we propose a novel framework for an invisible attack on PTMs with enhanced MD5 collision. The key idea is to generate two equal-size models with the same MD5 checksum by leveraging the MD5 chosen-prefix collision. Afterwards, the two ``same" models will be deployed on public websites to induce victims to download the poisoned model. Unlike conventional attacks on deep learning models, this new attack is flexible, covert, and model-independent. Additionally, we propose a simple defensive strategy for recognizing the MD5 chosen-prefix collision and provide a theoretical justification for its feasibility. We extensively validate the effectiveness and stealthiness of our proposed attack and defensive method on different models and data sets. The code of all experiments is available on a GitHub repository\footnote{\url{ https://github.com/xxxbrem/framework}}.
    \end{abstract}
    
    % \begin{IEEEkeywords}
    % component, formatting, style, styling, insert
    % \end{IEEEkeywords}

    \section{Introduction}
    \label{sec:intro}
    Recently, pre-trained models (PTMs) have achieved tremendous success in various tasks, especially for natural language processing \cite{devlin2019bert,radford2019language,liu2020roberta}, and computer vision \cite{bao2021beit,zhou2021ibot,kirillov2023segany}. However, the success of PTMs is highly dependent on massive amounts of data sets \cite{ouyang2022training,saharia2022photorealistic}, which is a non-trivial task to train by ourselves directly. As a result, one mainstream procedure is to download the third-party PTMs from some public sources, and then fine-tune the model on small task-specific data sets. While such a typical paradigm saves time and money, it also aggravates the security risks of attacks (e.g., backdoor attacks or data poisoning).
    
    Currently, there is widespread interest in the security of the deep learning model, such as face recognition \cite{dong2019efficient,jia2022adv} and autopilot \cite{eykholt2018robust,duan2021adversarial}. Although existing attacking methods such as data poisoning \cite{shafahi2018poison,huang2020metapoison,geiping2020witches} and backdoor attacks \cite{saha2020hidden,shen2021backdoor,doan2021lira} have achieved significant success, they usually ignore the stealthiness of attacks. They build upon the assumption that the attacker can invade victims' computers, which may not always hold in real-world scenarios, particularly in some private network environments. In other words, existing methods may fail to carry out an effective attack in a realistic application. 
    
    \begin{figure*}[t!]
    \begin{center}
    \includegraphics[width=1\textwidth]{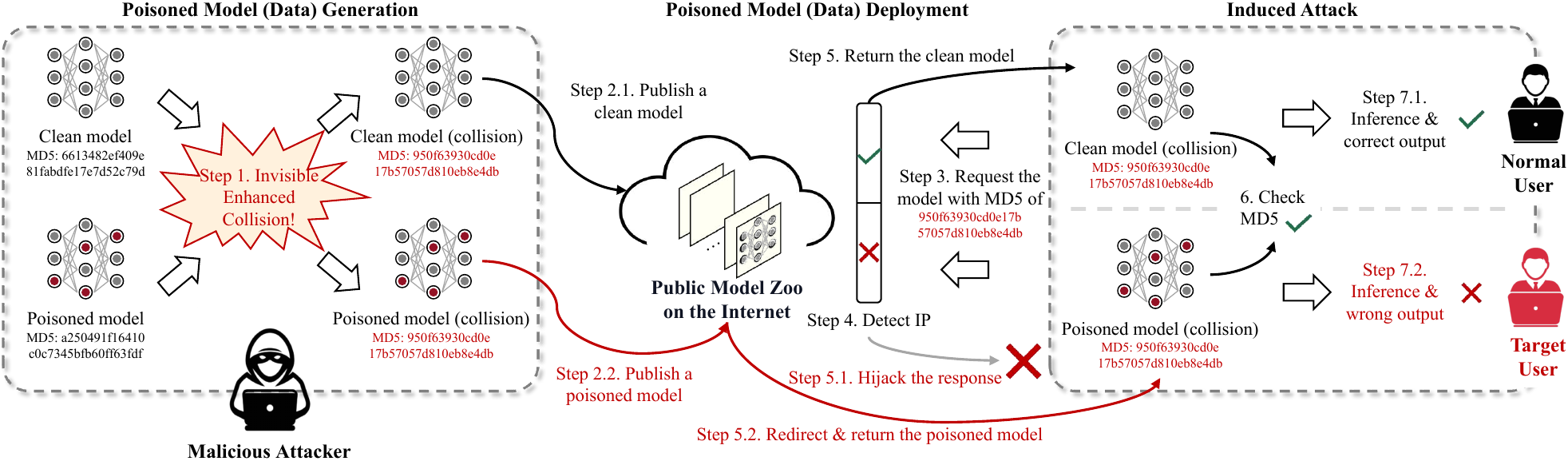}
    \end{center}
    \caption{The proposed framework for the invisible attack on PTMs with enhanced MD5 collision. It mainly consists of three components: (1) Poisoned Model (Data) Generation; (2) Poisoned Model (Data) Deployment; (3) Induced Attack.}
    \label{fig:framework}
    \end{figure*}
    
    Cryptographic hash functions are crucial components in many information security systems and are used for various purposes such as building digital signature schemes, message authentication codes or password hashing functions \cite{yang2017research,sharma2019cryptography}. Among them, the MD5 message-digest algorithm \cite{standardnational} is a widely used hash function that produces a 128-bit hash value, commonly serving as a checksum to verify data integrity against unintentional corruption. To protect the integrity of PTMs, the common way is to utilize it for hash-based commitment during downloading, which thus has potential pitfalls of being attacked. In this paper, we focus on the attack on PTMs with enhanced MD5 collision. A collision attack on a cryptographic hash attempts to find two inputs producing the same hash value (i.e., a hash collision). Specifically, our attack aims to yield two models (or data sets) based on the chosen-prefix collision \cite{collisions}, where the attacker allows choosing two arbitrarily different models (i.e., one can be a clean model, and the other is a poisoned model), and then append different contents to the end of a pair of files that result in the entire models having an equal hash value. In contrast to identical-prefix collision attack \cite{md5break}, chosen-prefix collision attacks are much more dangerous, and have already been demonstrated for MD5, having a significant impact in practice \cite{leurent2019collisions}. In addition, although chosen-prefix collision attacks often produce two equal-size models with the same MD5 value, the model size is slightly different from its original model. Therefore, we further propose an enhanced MD5 collision attack to produce the same original size, making the resulting poisoned model have ``no difference" from the clean model when victims download it.
    
    To better illustrate the basic idea, Figure \ref{fig:framework} gives our proposed framework for the invisible attack on PTMs with enhanced MD5 collision. The attack mainly consists of three phases. First, two equal-size models with the same MD5 hash values (i.e., clean model and poisoned model) are generated (cf. Section \ref{subsec:model}). Afterwards, the two models are deployed on public resources (e.g., some websites that all users can easily visit). Finally, the invisible and specific attack starts with a user downloading request, where normal users will get the clean models, while victims are unconscious of receiving the poisoned ones. 
    
    Building upon the enhanced MD5 collision, our attack method on PTMs has several desirable properties, which are summarized as follows:
    \begin{itemize}
        \item To the best of our knowledge, it is the first time to consider the MD5 chosen-prefix attack on large-scale PTMs and unlike traditional backdoor attacks or data poisoning methods, our proposed method is covert since the poisoned model has ``no difference" from the clean model (with the same MD5 hash value as well as the original model size),  making it difficult to identify in real-world scenarios.
        \item Our attack method is flexible and model-independent. It is not designed for a specific AI model or data set. Instead, it offers a general framework for the invisible attack on PTMs. 
        \item Additionally, we introduce a simple, effective, and general defensive strategy that tentatively recognizes the MD5 chosen-prefix collision within a data-driven learning framework. We provide a theoretical justification for its feasibility and present its ideal performance in defence effectiveness experiments.
        \item Extensive experiments have demonstrated the effectiveness and stealthiness of the proposed attack, which not only works on all AI models but also different modalities of data sets (e.g., images or texts).
    \end{itemize}

    \section{Preliminary}
    
    \subsection{Hash Functions}
    
    A reliable hash function should possess the following properties \cite{collisions}:
    \begin{itemize}
        \item \textbf{Pre-image resistance}: Given a hash value $h$, it should be computationally difficult to find any message $m$ that corresponds to a given hash value $h$ (i.e. $h = H(m)$).
        \item \textbf{Second pre-image resistance}: Given a message $m_1$, it should be computationally difficult to find another message $m_2 \neq m_1$  that produces the same hash value as $m_1$ (i.e. $H(m_1) = H(m_2)$).
        \item \textbf{Collision resistance}: It should be computationally difficult to find two different messages $m_1$ and $m_2$ that produce the same hash value (i.e. $H(m_1) = H(m_2)$).
    \end{itemize}

    Second pre-image resistance and collision resistance are also referred to as weak and strong collision resistance, respectively. Due to the larger domain of a hash function (which can even be infinite) compared to its range, the pigeonhole principle suggests the existence of numerous collisions. A brute force attack can find a pre-image or second pre-image for a hash function with n-bit hashes in approximately $2^n$ hash operations. With the birthday paradox, a brute force approach to generating collisions will succeed in around $2^{n/2}$ hash operations. Any attack requiring fewer hash operations than a brute force attack is considered a cryptographic hash function break.
    
    While widely used hash functions like MD5 and SHA-1 were once deemed secure against pre-image attacks, recent studies have highlighted their vulnerability to collision attacks \cite{md5break,stevens2017first,leurent2020sha}. Consequently, our focus has shifted towards the security implications of employing MD5 in certifying AI models, as its susceptibility to collision attacks can jeopardize the security of these models.

    \begin{figure*}[t!]
    \begin{center}
    \includegraphics[width=1\textwidth]{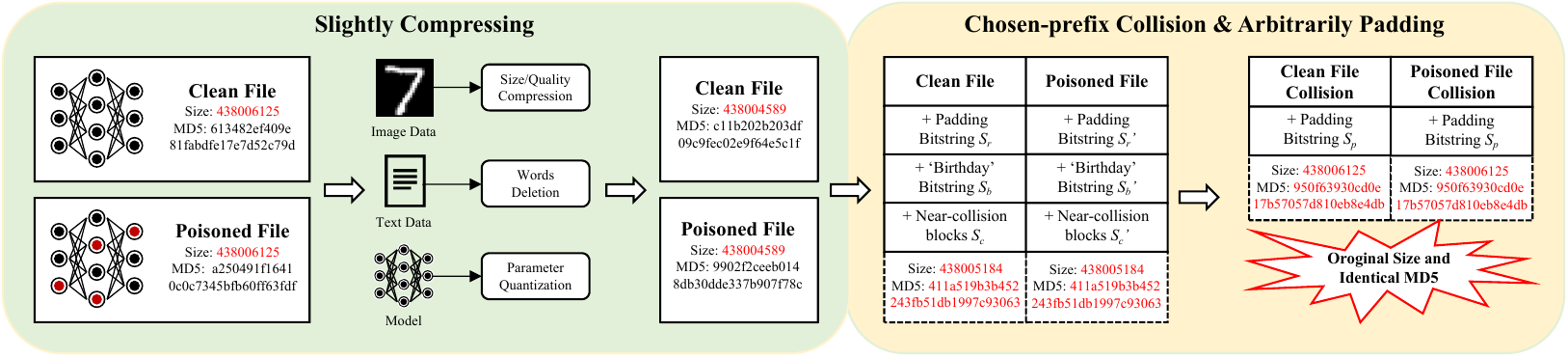}
    \end{center}
    \caption{The pipeline of constructing collisions with the same size as the source file. It mainly consists of $4$ steps: (1) The Choice of Attacking Methods. We generate a pair of clean and poisoned files (models or data sets) by existing deep learning attack methods. (Whether the size is the same or not) (2) Slightly Compressing. According to the type of files, we compress them to ensure the size of their collision versions does not exceed the size of the source file. (3) Chosen-prefix Collision. We generate $2$ collision files with sizes still not exceeding the original file. (4) Arbitrarily Padding. We add contents after the collision file until its size matches that of the source file.}
    \label{fig:resize}
    \end{figure*}
    
    \subsection{Chosen-prefix Collision}
    
    There are two main forms of collision attacks: identical prefix collision (IPC) attacks and chosen-prefix collision (CPC) attacks. In an IPC attack, the attack algorithm finds two different suffixes $S$ and $S'$ that, when concatenated with a given prefix $P$ (usually a message needing certification), produce the same MD5 hash:
    
    \begin{align}
        {\text{MD5}}(P||S) = {\text{MD5}}(P||S').
    \end{align}
    
    IPC attacks are considered less dangerous since they do not alter the original message and only add meaningless strings to the suffixes. However, based on IPC attacks, Stevens et al. \cite{collisions} demonstrated a CPC algorithm that poses a nonnegligible threat to real applications. In a CPC attack, given two arbitrary prefixes $P$ and $P'$, the attack algorithm constructs two different suffixes $S$ and $S'$ that, when concatenated with the respective prefixes, produce the same MD5 hash:
    
    \begin{align}
        {\text{MD5}}(P||S) = {\text{MD5}}(P'||S').
    \end{align}
    
    The suffixes $S$ and $S'$ are constructed as follows \cite{collisions}.
    \begin{itemize}
        \item\textbf{Padding Bitstrings $S_r$ and $S_r'$}: $S_r$ and $S_r'$ are designed to make the lengths of $P||S_r$ and $P'||S_r'$ equal to $512n - 64 - k$, where $n$ is a positive integer, and $0 < k < 32$ is a parameter.
        
        \item\textbf{`Birthday' Bitstrings $S_b$ and $S_b'$}: The birthday bitstrings $S_b$ and $S_b'$ have a bitlength of $64 + k$ and are chosen in a way that applying the MD5 compression function to $P||S_r||S_b$ and $P'||S_r'||S_b'$ results in $\text{IHV}_n$ and $\text{IHV}_0$, respectively.

        To hash a message consisting of $N$ blocks, MD5 goes through $N + 1$ states $\text{IHV}_i$, for $0 \leq i \leq N$, known as intermediate hash values. Each intermediate hash value $\text{IHV}_i$ consists of four 32-bit words: $a_i$, $b_i$, $c_i$, $d_i$. For $i = 0$, these values are fixed public constants. For $i = 1, 2, \ldots, N$, the intermediate hash value $\text{IHV}_i$ is computed using the MD5 compression function:
        
        \begin{equation}
            \text{IHV}_i = \text{MD5Compress}(\text{IHV}_{i-1}, M_i)
        \end{equation}
    
        The goal of the birthday search is to find suitable birthday bitstrings, $S_b$ and $S_b'$, that eliminate the differences in the MD5 hash values. The objective is to achieve a differential vector of the form $\delta\text{IHV}_n = (0, \delta b, \delta c, \delta c)$. The process of determining such $S_b$ and $S_b'$ requires approximately $\sqrt{\pi}2^{32+(k/2)}$ calls to the MD5 compression function. This search space has 64 bits, and collisions can be found with a relatively low computational cost.
        
        \item\textbf{Near-collision Blocks Bitstrings $S_c$ and $S_c'$}: The idea is to eliminate the difference $\delta\text{IHV}n$ in $r$ consecutive steps, where $r$ is a chosen value. This is accomplished by constructing $S_c = S_{c,1}||S_{c,2}||\ldots||S_{c,r}$ and $S_c' = S_{c,1}'||S_{c,2}'||\ldots||S_{c,r}'$ for $r$ pairs of near-collision blocks ($S_{c,j}$, $S_{c,j}'$) with $1 \leq j \leq r$. For each pair of near-collision blocks, a differential path is created such that the difference vector $\delta\text{IHV}_{n+j}$ has a lower weight than $\delta\text{IHV}_{n+j-1}$. This process continues until after $r$ pairs, resulting in $\delta\text{IHV}_{n+r} = (0, 0, 0, 0)$.
        
        \item\textbf{Achieving Collision:} Finally, the collision is achieved by concatenating all these sub-suffixes:
        
        \begin{align}
            {\text{MD5}}(P||S_r||S_b||S_c) = {\text{MD5}}(P'||S_r'||S_b'||S_c').
        \label{equation:cpc}
        \end{align}
    \end{itemize}
    
    \section{Methodology}

    \subsection{Invisible Attack}

    \subsubsection{Enhanced MD5 Collision}
    \label{subsec:model}
    
    MD5 collision has found applications in various software domains such as colliding documents, hash-based commitments, content-addressed storage, and file integrity checking \cite{cpcApplication}. However, a notable drawback of MD5 collision is that the collided size is slightly larger than the original, as indicated in Equation \ref{equation:cpc}, which can potentially raise suspicion and lead to detection.
    
    In this section, we delve into the process of generating poisoned models and data, as illustrated in Figure \ref{fig:framework}. To ensure the stealthiness of the attacks, we propose an enhanced MD5 collision technique.
    
    The core concept involves manipulating the size of the data and models through compression, deletion, and quantization techniques. By employing chosen-prefix collision (CPC), we generate two distinct files and pad them to match the exact size of the original file, as depicted in Figure \ref{fig:resize}. This approach enhances the covert nature of the collision attack, minimizing the likelihood of detection.
    
    \textbf{The Choice of Attacking Methods.}
    Depending on the data type, we generate two files through backdoor attacks targeting PTMs or data poisoning targeting the training dataset. Both files exhibit similar performance with normal inputs but may show variations with specific inputs. The size of the generated files is insignificant since they are padded to the same length during the CPC process. Notably, their MD5 checksums differ.
    
    \textbf{Slightly Compressing.}
    To achieve equal sizes between the collision file and the original one, we employ compression techniques that provide additional storage capacity for collision contents while minimizing the impact on testing accuracy. One effective method for models is parameter quantization. In Figure \ref{fig:resize}, we demonstrate the process of converting a small portion of network parameters from 32-bit precision floating-point to 16-bit, resulting in a reduction of 1536 bytes. When dealing with image data, we can reduce the image's size or quality, while for text data, we can eliminate less significant words. These techniques effectively decrease file size while having minimal impact on final accuracy. It is important to note that in most cases, the contents added by the CPC are less than 1 KB, thus the compression level should exceed 1 KB in practice.
    
    \textbf{Chosen-prefix Collision and Padding.}
    As shown in Equation \ref{equation:cpc}, both of the two files undergo an appending process that includes the padding bitstring $S_r$, the 'birthday' bitstring $S_b$, and the near-collision blocks $S_c$. Despite the compression performed in the preceding step, the size of these files remains smaller compared to the original source file. To ensure consistency, we add an arbitrarily selected suffix to the end of both files, thereby making their sizes identical to the original file. Importantly, due to the inherent properties of MD5, their MD5 checksums remain equivalent.
    
    \subsubsection{General MD5 Attack Framework}
    
    Based on the method described in Section \ref{subsec:model}, We propose a general MD5 CPC attack framework that specifically focuses on deep learning models and data sets. 
    
    PTMs have emerged as the standard approach for numerous AI applications. Nonetheless, not all AI users possess the capability to train their own models, leading them to rely on downloading PTMs and data sets from online platforms and utilising insecure MD5 checksums to verify file integrity. In this study, we present a workflow showcasing an enhanced MD5 collision attack in this particular scenario.

    We make an assumption that the attack remains concealed and targets specific victims exclusively. Furthermore, we assume that the attacker possesses the capability to identify the intended victims and exert control over their network environment. For instance, the attacker may achieve this by altering routers or engaging in DNS pollution, thereby redirecting the download address of a model.
    
    The complete workflow is illustrated in Figure \ref{fig:framework}, which comprises seven steps:
    \begin{itemize}
        \item\textbf{Step $1$: Invisible Enhanced Collision}: The attacker generates two collisions $\text{Col}(C)$, $\text{Col}(P)$ with the same MD5 and size, based on enhanced Md5 collision attack. These collisions are created from a clean file $C$ and a poisoned file $P$ (model or data).
        \item\textbf{Step $2$: Publishing}: The attacker releases the clean file $\text{Col}(C)$ and the poisoned file $\text{Col}(P)$, along with their MD5 checksum, on a model zoo or similar platform for downloading and file integrity checking. Due to the collision, both $\text{Col}(C)$ and $\text{Col}(P)$ have the same MD5 value, denoted as $\text{MD5}(\text{Col}(C)) = \text{MD5}(\text{Col}(P))$.
        \item \textbf{Step 3: User Request}: Regular users, as well as the target individuals, send downloading requests with an expected MD5 value $h$.
        \item\textbf{Step $4$: IP Detection}: The website's controller checks whether the IP address belongs to the target person or a normal user.
        \item\textbf{Step $5$: File Return}: For normal users with $\text{IP}(N)$, the attack responds by providing the clean file $\text{Col}(C)$, while for target users with $\text{IP}(T)$, the attack hijacks the response according to IP, and returns the poisoned model $\text{Col}(P)$ instead. 
        \item\textbf{Step $6$: MD5 Check}: All users verify the MD5 checksum of the downloaded file. Since the collision occurred during Step 1, it is guaranteed that $h = \text{MD5}(\text{Col}(C)) = \text{MD5}(\text{Col}(P))$.
        \item\textbf{Step $7$: Inference Attack}: Normal users are unaware of any abnormalities and only access the clean model or data to generate accurate predictions $\hat{y_c}$. In contrast, the target individuals unknowingly encounter the hidden attack through the poisoned model, resulting in incorrect outputs $\hat{y_p}$. This discrepancy arises because the MD5 checksum cannot differentiate between the two files, leading to an induced attack.
    \end{itemize}
    
    \begin{algorithm}[tb]
        \caption{General MD5 Attack Framework}
        \label{alg:framework}
        \raggedright\textbf{Input}: A clean file $C$, a poisoned file $P$, normal users $N$, target users $T$, downloading file's MD5 $h$, testing data $D$\\
        \raggedright\textbf{Output}: Normal users' and victims' prediction $\hat{y_c}$, $\hat{y_p}$\\
        \begin{algorithmic}[1] %[1] enables line numbers
            \STATE Enhanced collision attack: clean collision: $\text{Col}(C)$, poisoned collision: $\text{Col}(P)$\\
            $s.t. \text{MD5}(\text{Col}(C)) == \text{MD5}(\text{Col}(P))$
            \STATE Normal IP: $\text{IP}(N)$, target IP: $\text{IP}(T)$, current IP: $ip$
            \IF{$ip == \text{IP}(N)$}
                \STATE Return clean collision $\text{Col}(C)$ to normal users $N$
            \ELSE
                \IF{$ip == \text{IP}(T)$}
                    \STATE Return poisoned collision $\text{Col}(P)$ to target users $T$ 
                \ENDIF
            \ENDIF
            \STATE Use $\text{Col}(C)$ and $\text{Col}(P)$ to get a clean model $\Pi_c$ and a poisoned model $\Pi_p$
            \STATE Make predictions: $\hat{y_c} = \Pi_c(D)$, $\hat{y_p} = \Pi_p(D)$
            \STATE \textbf{return} $\hat{y_c}$, $\hat{y_p}$
        \end{algorithmic}
    \end{algorithm}
    
    Finally, the overall workflow is given in Algorithm \ref{alg:framework}.
    
    \subsection{Defence}

    Currently, several software-based defensive methods have been developed to counter collisions. One straightforward approach is to directly compare two files. However, comparing two models becomes challenging due to parameter changes during fine-tuning. Additionally, acquiring the original file poses a difficulty, as the attacker only provides the manipulated model or dataset.

    We made a crucial observation that the pattern of regular files (models or data sets) differs significantly from the pattern of collision parts, as depicted in Figure \ref{fig:JS}. In other words, clean samples bear little resemblance to collision samples and collision samples differ greatly from one another. 

    Another defensive method involves detecting the last near-collision block of an attack and leveraging key insights from hash cryptanalysis \cite{stevens2013counter}. However, this method is specific to a particular hash function and requires knowledge about collision generation. Therefore, we propose training a deep learning model that is general, straightforward, and transferable. This model can effectively recognize the collision parts and provide defence against collision attacks.
    
    In a realistic scenario, users typically do not have access to the original file. Hence, they train a model on similar data and utilize it to detect collisions, thus defending against such attacks. Since testing a file involves examining numerous samples, it can be resource-consuming. To address this, we introduce Jaccard Similarity (JS) as an efficient means to filter samples and accelerate the testing phase.
    
    \subsubsection{Theoretical Justification for Pattern Discrepancy}
    \label{subsec:justify}
        
    In this section, we leverage the theory of the birthday problem to theoretically justify our observation: Under certain circumstances, the probability of collision samples sharing the same characteristics (birthday probability) between collision samples only, or between collision samples and clean samples, is relatively low, while the birthday probability between clean samples is higher.
    
    \begin{lemma}\label{lemma:birthday}
        (The Birthday Problem). In the birthday search space $S$, and given a set of $N$ randomly chosen people, the probability $P(A)$ of at least two sharing a birthday is approximately $1-e^{-N^2/2S}$.
        \begin{proof}
            The probability that person $i$ is born on a different day than person $1, 2, \ldots, i-1$:
            \begin{equation}
                \begin{aligned}
                    P(E_i) = \frac{S-(i-1)}{S}.
                \end{aligned}
            \end{equation}
            The probability that no $2$ people share a birthday:
            \begin{equation}\label{equation:overlineA}
                \begin{aligned}
                    P(\overline{A}) = \prod_{i=1}^{N}P(E_i) = \prod_{i=1}^{N-1}(1-\frac{i}{S}).
                \end{aligned}
            \end{equation}   
            Based on the Taylor expansion of $e^x$:
            \begin{equation}
                \begin{aligned}
                    e^x = \sum_{N=0}^{\infty} \frac{x^N}{N!}.
                \end{aligned}
            \end{equation}          
            When $x \ll 1$:
            \begin{equation}
                \begin{aligned}
                    e^x \approx 1 + x + \frac{x^2}{2!} + \frac{x^3}{3!} + \cdots \approx 1 + x.
                \end{aligned}
            \end{equation}   
            Thus, the approximation of $P(A)$:
            \begin{equation}
                \begin{aligned}
                    P(A) = 1- \prod_{i=1}^{N-1}e^{-\frac{i}{S}} = 1 - e^{\frac{-N(N-1)}{2S}} \approx 1 - e^{\frac{-N^2}{2S}}.
                \end{aligned}
            \end{equation}         
        \end{proof}
    \end{lemma}
    
    The comparison of clean and collision samples could be viewed as a birthday problem. In this analogy, the search space for birthdays corresponds to the vocabulary of tokens, while the number of individuals relates to the sequence length of each sample. By applying the concept of Lemma \ref{lemma:birthday}, we can estimate the probability of two individuals (tokens in samples) sharing the same ``birthday" (token in the vocabulary) among clean samples, collision samples, and both types of samples combined.
    
    \begin{theorem}\label{theorem: thm1}
        Let the vocabulary size of clean samples be $S_a$, and the vocabulary size of collision samples be $S_b$ ($S_a \ll S_b < S$). Suppose the proportion of clean samples is $p_a$, and the proportion of collision samples is $p_b$ ($p_a + p_b = 1$), then the probability $P(A_{clean})$ of at least two tokens sharing a birthday is approximately $1-e^{-p_a^2N^2/2S_a}$. Similarly, the probability $P(A_{collision})$ of at least two tokens sharing a birthday is approximately $1-e^{-p_b^2N^2/2S_b}$.
    \end{theorem} 
    
    If we consider the vocabulary size to be $S^2$, according to Theorem \ref{theorem: thm1}, the probability of at least two tokens sharing a birthday in collision samples is $P(A') = 1-e^{-N'^2/2S^2}$, where $N' = N^2$ is the minimum number of samples required to maintain the same birthday probability. This implies that there is more space to differentiate between clean and collision samples with the increased vocabulary size.
    
    \begin{theorem}\label{theorem: thm2}
        Suppose the proportion of clean samples is $p_a$, and the proportion of collision samples is $p_b$ ($p_a + p_b = 1$ and $p_a \gg p_b$), then the probability $P(A_{clean\_collision})$ of at least two tokens sharing a birthday is approximately $1 - e^{\frac{-N^2}{2S}} + e^{\frac{-p_a^2N^2}{2S_a}} + e^{\frac{-p_b^2N^2}{2S_b}}$.
    \end{theorem} 
    
    \begin{proof}
        Let $E_{sample}$: no $2$ samples share a birthday, $E_{clean}$: no $2$ clean samples share a birthday, $E_{collision}$: no $2$ collision samples share a birthday, $E_{clean\_collision}$: no $2$ clean and collision samples share a birthday. 
        \begin{equation}
            P(E_{sample}) = P(E_{clean}) + P(E_{collision}) + P(E_{clean\_collision})
        \end{equation}
        According to Equation \ref{equation:overlineA}: 
        \begin{equation}
            P(E_{sample}) = \prod_{i=1}^{N-1}(1-\frac{i}{S}) \approx e^{\frac{-N^2}{2S}}.
        \end{equation}
        According to Theorem \ref{theorem: thm1}:
        \begin{gather}\label{equation:11}
            P(E_{clean}) = \prod_{i=1}^{p_aN-1}(1-\frac{i}{S_a}) \approx e^{\frac{-p_a^2N^2}{2S_a}}, \\
            P(E_{collision}) = \prod_{i=1}^{p_bN-1}(1-\frac{i}{S_b}) \approx e^{\frac{-p_b^2N^2}{2S_b}}.
        \end{gather}
        Thus, 
        \begin{gather}
            P(A_{clean\_collision}) = 1 - P(E_{clean\_collision}) \nonumber\\
            \approx 1 - e^{\frac{-N^2}{2S}} + e^{\frac{-p_a^2N^2}{2S_a}} + e^{\frac{-p_b^2N^2}{2S_b}}.\label{equation:13}
        \end{gather}
    \end{proof}

    \begin{figure*}[t!]
    \begin{center}
    \includegraphics[width=1\textwidth]{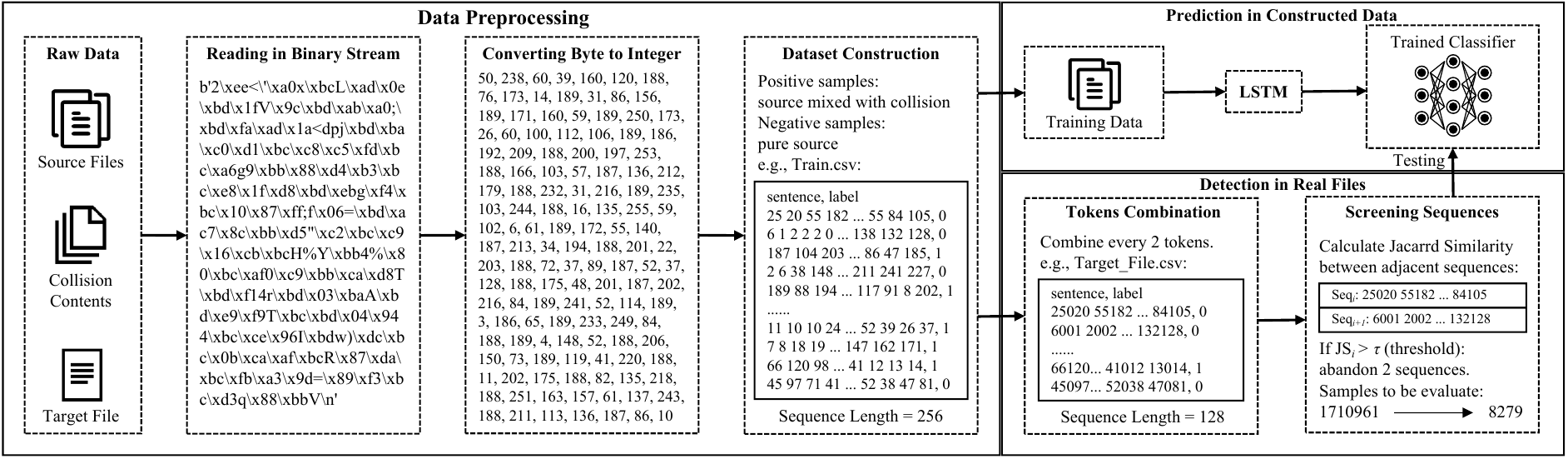}
    \end{center}
    \caption{The outline of MD5 collision recognition. It contains three parts: (1) Data Preprocessing. We read the raw data in the binary stream and convert it to integer tokens. Similar to the sentiment data sets, we consider tokens in a fixed length as a sentence and assign a label to each one. (2) Prediction. We train a deep model to classify collisions. (3) Detection. In the realistic data, we introduce an unsupervised approach to help classify and improve accuracy and efficiency.}
    \label{fig:defence}
    \end{figure*}
    
    \begin{theorem}\label{theorem: thm3}
        (Collision Pattern Discrepancy Theorem). The birthday probability between collision samples $P(A_{collision})$ and the birthday probability between clean and collision samples $P(A_{clean\_collision})$ are close, and both much less than the birthday probability between clean samples $P(A_{clean})$, such that:
        \begin{equation}
            P(A_{collision}) \approx P(A_{clean\_collision}) \ll P(A_{clean}).
        \end{equation}
        \begin{proof}
            According to Theorem \ref{theorem: thm1} and Theorem \ref{theorem: thm2}, $S_a \ll S_b < S$ and $p_a \gg p_b$.
            In this way,
            \begin{equation}
                \lim\limits_{\frac{S_b}{S_a} \rightarrow +\infty \atop \frac{p_a}{p_b} \rightarrow +\infty}\frac{\frac{-p_a^2N^2}{2S_a}}{\frac{-p_b^2N^2}{2S_b}} = \lim\limits_{\frac{S_b}{S_a} \rightarrow +\infty \atop \frac{p_a}{p_b} \rightarrow +\infty}\frac{S_b}{S_a}(\frac{p_a}{p_b})^2 \rightarrow +\infty.
            \end{equation}
            Then, according to Equation \ref{equation:11},
            \begin{equation}
                P(E_{clean}) \approx e^{\frac{-p_a^2N^2}{2S_a}} \ll e^{\frac{-p_b^2N^2}{2S_b}} \approx P(E_{collision})
            \end{equation}
            Thus,            
            \begin{gather}
                1 \approx 1 - P(E_{clean}) = P(A_{clean}) \gg \nonumber\\
                P(A_{collision}) = 1 - P(E_{collision}) \approx 0
            \end{gather}
            Let $p_{S_b} = \frac{S_b}{S}$, such that
            \begin{equation}
                \ln2p_b^2 = p_{S_b}.
            \end{equation}
            Then,
            \begin{equation}\label{equation:19}
                e^{\frac{-p_a^2N^2}{2S_a}} + 2e^{\frac{-p_b^2N^2}{2S_b}} = e^{\frac{N^2}{2S}}.
            \end{equation}
            Thus, according to Equation \ref{equation:11}, Equation \ref{equation:13} and Equation \ref{equation:19},
            \begin{equation}
                P(A_{collision}) \approx P(A_{clean\_collision}) \ll P(A_{clean}).
            \end{equation}
        \end{proof}
    \end{theorem}
    
    By applying the principles of the birthday problem, we substantiate our observation that the pattern of collision samples differs significantly from both other collision samples and clean samples. This theoretical analysis lends support to our proposed approach of training a deep learning model to effectively recognize the distinctive features of collision parts and provide defence against collision attacks via the JS method.

    \begin{algorithm}[tb]
        \caption{Classifier for Discriminating Collisions}
        \label{alg:cct}
        \raggedright\textbf{Input}: Clean files $S$(source), $T$(target), and their collision parts: $\text{IPC}(S)$, $\text{CPC}(T)$\\
        \raggedright\textbf{Output}: A trained classifier $\Pi_\theta$, prediction results: $\hat{y}$\\
        \begin{algorithmic}[1] %[1] enables line numbers
            \STATE Read $S$, $T$, $\text{IPC}(S)$, $\text{CPC}(T)$ in the binary stream
            \STATE Convert bytes to integers
            \STATE Generate training data $Tr$:\\
            positive samples $s + i$: $s \in S$ and $i \in \text{IPC}(S)$,\\
            negative samples $s'$: $s' \in S$, $s \cap s' = \emptyset$.
            \STATE Generate testing data $Te$:\\
            positive samples $t + c$: $t \in T$ and $c \in \text{CPC}(T)$,\\
            negative samples $t'$: $t' \in T$, $t \cap t' = \emptyset$.
            \FOR{$i=1$ {\bfseries to} $epoch$}
                \STATE $L = Cross$-$Entropy(\Pi_{\theta_i}(Tr_x), Tr_y)$
            	\STATE Update weights ${\theta_{i+1}} \leftarrow \theta_i + \eta \cdot \nabla_L$
            \ENDFOR
            \STATE Prediction $\hat{y} = \Pi_\theta(Te)$
            \STATE \textbf{return} $\Pi_\theta$
        \end{algorithmic}
    \end{algorithm}

    \subsubsection{Transferable Defence}
    
    In this section, we utilize MD5 collisions as an example to explore the detection of collisions in regular files. We train a model on similar training datasets and perform transfer learning tasks to effectively detect collisions. The fundamental workflow of our defence method is illustrated in Figure \ref{fig:defence}.
    
    In order to achieve a collision, CPC adds specific suffixes to the source file. It is worth noting that these suffixes do not have to be added at the end of the source file. Notably, these suffixes can be inserted at any position within the file, rather than just at the end. This positioning flexibility poses a challenge for detection. However, these collision suffixes exhibit discernible patterns that differ from normal file contents, according to Section \ref{subsec:justify}. This distinction allows us to train a neural network to differentiate between regular file contents and collision suffixes.
    
    We frame the detection of collision attacks as a sequence classification task and train an LSTM model \cite{hochreiter1997long} for detection. Specifically, we generate the training data set by dividing the source file and the collision suffixes into 256-length byte sequences as negative samples and positive samples, respectively. Algorithm \ref{alg:cct} outlines how to deal with data and train a collision classifier. Due to the time-consuming process of generating CPC suffixes, we use IPC to generate collision suffixes, as they share similar patterns. For positive samples, half of the sequence bytes are sampled from collision suffixes, and the other half are meaningful bytes sampled from the source file. This setting stimulates the real scenario that collision suffixes might be split in random positions.
    
    \begin{algorithm}[tb]
        \caption{Collision Detection in Files}
        \label{alg:cd}
        \raggedright\textbf{Input}: A file $F$ to be detected, trained model $\Pi_\theta$\\
        \raggedright\textbf{Output}: Prediction $\hat{y}$
        \begin{algorithmic}[1] %[1] enables line numbers
            \STATE Convert $F$ to integers of the byte stream.
            \STATE Initialize potential collisions set: $P = \emptyset$, threshold $\tau$ = 0. 
            \FOR{$i=1$ {\bfseries to} $len(F) - 1$}
            	\STATE Similarity $JS_i = Jaccard(F_i, F_{i+1})$ 
            	\IF{$JS_i \leq \tau$ and $JS_{i+1} \leq \tau$}
            	    \STATE $P$ append $F_{i-1}, F_{i}, F_{i+1}$
            	\ENDIF
            \ENDFOR
            \STATE $\hat{y} = \Pi_\theta(Set(P))$
            \STATE \textbf{return} $\hat{y}$
        \end{algorithmic}
    \end{algorithm}
    
    We now turn our attention to the problem of collision detection in realistic scenarios, where model files like BERT can consist of hundreds of millions of bytes. It is resource-consuming to split all the files into 256-length sequences and send them into the collision classifier. To address this issue, we calculate the Jaccard Similarity (JS) between sequences as initial filtering. Only samples with low JS are considered as candidate samples. The detection algorithm is presented in Algorithm \ref{alg:cd}. 

    \begin{figure}[!t]
        \begin{center}
            \includegraphics[width=0.46\textwidth]{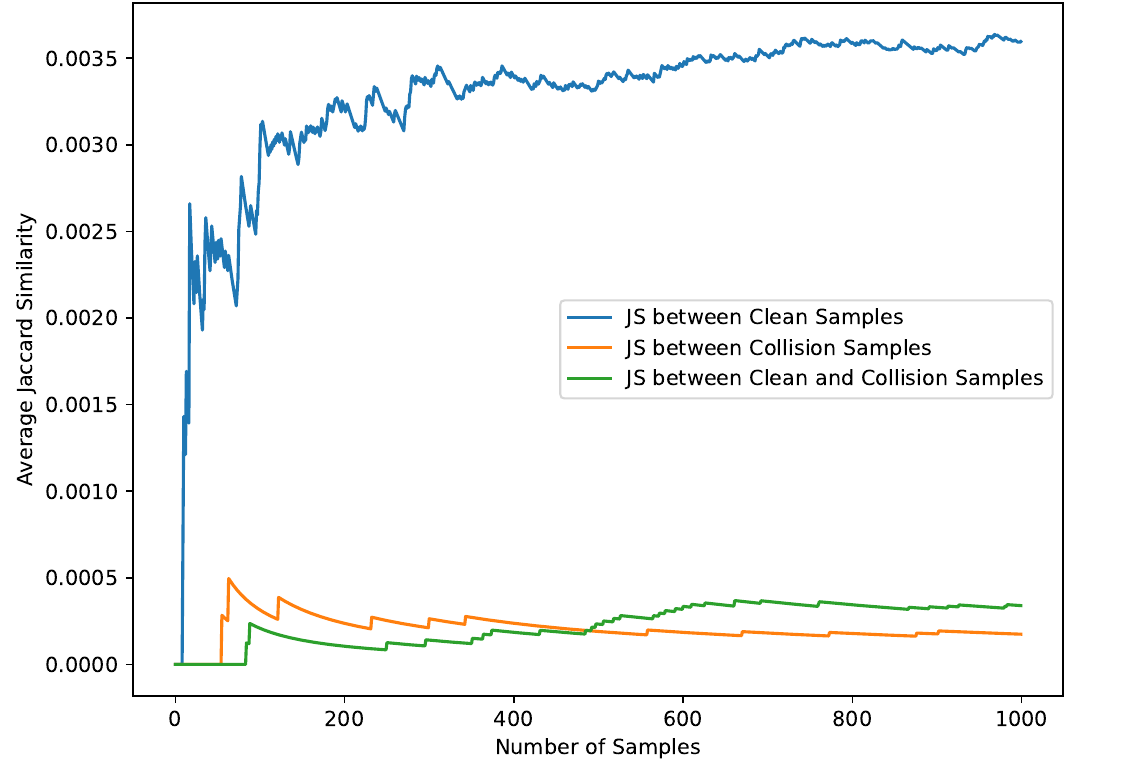}
        \end{center}
        \caption{Convergence of Jaccard Similarity (JS) in different sample types and sizes. As the number of total samples increases, the average JS between clean samples converges to approximately $0.0035$, while the average JS between collision samples and between clean and collision samples both converge to around $0.00025$. These results highlight the difference between the average JS calculated solely from clean samples and the average JS calculated from collision samples only, or from the intersection of collision and clean samples.}
        \label{fig:JS}
    \end{figure}
    
    To distinguish collisions from meaningful contents, we combine every two bytes of the binary stream into a new byte as an integer, enlarging the vocabulary size to $256^2$, due to Theorem \ref{theorem: thm1}. We then employ a time complexity of $O(n)$ to compare adjacent sequences, and we consider a pair of sequences and their neighbors as potential collisions if two successive JS results are less than or equal to a small value, $\tau$. As shown in Figure \ref{fig:JS}, the average JS calculated solely from clean samples differs significantly from the average JS calculated from collision samples only or from the intersection of collision and clean samples, with $seq\_length = 64$ and $\tau = 0$. This technique ensures that collisions have little in common with themselves or meaningful contents, whereas meaningful contents resemble themselves relatively. 
    
    After applying this process, we are left with a reduced set of samples, which we then classify using our trained model for collision detection. This approach enables us to detect collisions in large files with higher accuracy and efficiency.

    \section{Experiments}

    \begin{table*}
        \centering
                \caption{Enhanced MD5 collision attacks via RIPPLe and Gradient Matching. For model cases, We train a clean bert-base-uncased model (BERT) on the SST-2 data set and apply RIPPLe to poison the model (BERT + Backdoor). We resize both models, generate their collision versions (BERT + Coll and BERT + Backdoor + Coll), and test them by the clean and poisoned data. We record the MD5 checksum and size of each file for comparison. Similarly, for data cases, we use Gradient Matching to generate the clean and poisoned data sets (CIFAR and CIFAR Poisoned). After resizing the data sets, we generate their collision versions (CIFAR + Coll and CIFAR Poisoned + Coll). Then, we train clean and poisoned models from scratch on four different data sets using the ResNet18 model. We test the models with clean data and triggers. We also record the MD5 and size of each file.}
        \label{table:RIPPLe}
        \begin{tabular}{cccccc}
            \toprule
            & File Type & MD5 Checksum & File Size (Bytes) & Clean Data Accuracy & Trigger Prediction\\
            \midrule
            \multirow{4}{*}{Model} & BERT & 6613482ef409e81fabdfe17e7d52c79d & \textbf{438006125} & $0.923$ & $0.051$ \\
            & BERT + Backdoor & a250491f16410c0c7345bfb60ff63fdf & \textbf{438006125} & $0.923$ & $1.000$ \\
            & BERT + Coll & \textbf{950f63930cd0e17b57057d810eb8e4db} & \textbf{438006125} & $0.923$ & $0.051$ \\
            & BERT + Backdoor + Coll & \textbf{950f63930cd0e17b57057d810eb8e4db} & \textbf{438006125} & $0.923$ & $1.000$ \\
            \midrule
            \multirow{4}{*}{Data} & CIFAR & c58f30108f718f92721af3b95e74349a & \textbf{170498071} & $0.917$ & autopilot \\
            & CIFAR Poisoned & 79f85ef6b1622e9a87adbdf522a3cf4e & 169587629 & $0.917$ & cat \\
            & CIFAR + Coll & \textbf{bb491cdeeabeaa3b12ddfd27ff0abf70} & \textbf{170498071} & $0.917$ & autopilot \\
            & CIFAR Poisoned + Coll & \textbf{bb491cdeeabeaa3b12ddfd27ff0abf70} & \textbf{170498071} & $0.917$ & cat \\
            \bottomrule
        \end{tabular}
    \end{table*}

    \subsection{Attack Evaluation}
    \textbf{Selection of Attacking Algorithms.} In this study, we assess two types of MD5 attacks: backdoor attacks and data poisoning attacks. To conduct backdoor attacks, we utilize RIPPLe \cite{RIPPLe} and apply it specifically to the BERT model \cite{BERT}. As for data poisoning attacks, we employ Gradient Matching \cite{geiping2020witches} to craft poisoned data for the CIFAR-10 dataset. We then train both clean and poisoned ResNet models \cite{he2016deep} in order to evaluate the efficacy of the data poisoning attack.
    
    \textbf{Data Sets.} Our experimentation involved utilizing the Stanford Sentiment Treebank (SST-2) data set \cite{socher2013recursive} for testing weight poisoning on sentiment classification in MD5 attacks through backdoor attacks. In order to conduct MD5 data poisoning attacks, we employed the CIFAR-10 data set \cite{krizhevsky2009learning}. Within this data set, we crafted poisoned data and subsequently evaluated the performance of the victim model.
    
    \textbf{Evaluation Metrics.} To further analyze enhanced collision attacks through backdoor attacks, we compare the MD5 checksums, file sizes, accuracy with clean testing files, and predictions from triggers across the clean model, poisoned model, and their collision counterparts. Similarly, for enhanced collision attacks using data poisoning, we train clean and poisoned models using the clean data and the selected label-poisoned data, as well as their corresponding collision counterparts. We then assess the MD5 checksums, file sizes, models' testing accuracy, and predictions generated from triggers.

    \textbf{Attack Effectiveness.}
    The results of enhanced MD5 collision attacks carried out via RIPPLe and Gradient Matching are presented in Table \ref{table:RIPPLe}.
    
    For backdoor attacks, we have the BERT model and the BERT + Backdoor model, which is fine-tuned on the SST-2 data set with its weights poisoned by RIPPLe. As these models are distinct, their MD5 checksums differ significantly. However, after the enhanced MD5 collision attacks, their MD5 checksums are transformed to 950f63930cd0e17b57057d810eb8e4db, and the file size remains the same at $438,006,125$.
    
    Regarding data poisoning attacks, the CIFAR Poisoned data set is generated using the Gradient Matching method, resulting in a distinct MD5 checksum and file size compared to the original data set. However, after the enhanced MD5 collision attacks, we obtain two additional data sets, CIFAR + Coll and CIFAR Poisoned + Coll, both having the same MD5 checksum as bb491cdeeabeaa3b12ddfd27ff0abf70, and the same file size as the original data set at $170,498,071$. We proceed to train four models using ResNet based on these four different data sets and evaluate their performance using clean testing data and triggers. The results obtained from the testing on clean data and triggers demonstrate that the collision files can achieve an equivalent level of effectiveness.

    \subsection{Defence Evaluation}

    \begin{table}[!t]
        \centering
                \caption{Transfer Tasks: comparison of collision prediction accuracy. The training and testing files represent the sources of the training and testing data, respectively. The diagonal of each model shows the results of the domain experiment, which performs well across all file types. The other off-diagonal experiments show inconsistent results.}
        \label{table:prediction}
        \begin{tabular}{cccccccc}
            \toprule
            \multicolumn{2}{c}{} & \multicolumn{4}{c}{Testing File} \\
            \cmidrule(lr){3-6} 
            Model & Training File & BERT & ViT &  CIFAR-10 & SST-2 \\
            \midrule 
            \multirow{4}{*}{LSTM} & BERT & \textbf{1.00} & 1.00 & 0.57 & 0.81 \\
                    & ViT      & 1.00 & \textbf{1.00} & 0.55 & 0.93 \\
                    & CIFAR-10 & 0.50 & 0.50 & \textbf{0.99} & 1.00 \\
                    & SST-2    & 0.53 & 0.53 & 0.55 & \textbf{1.00} \\
            \midrule 
            \multirow{4}{*}{Bayes} & BERT & 0.98 & 0.98 & 0.50 & 0.50 \\
            & ViT & 0.98 & 0.97 & 0.50 & 0.50 \\
            & CIFAR-10 & 0.50 & 0.50 & 0.72 & 0.91 \\
            & SST-2 & 0.50 & 0.50 & 0.50 & 1.00 \\
            \bottomrule
        \end{tabular}
        \vspace{5pt}
    \end{table}

    \textbf{Data Sets.}    We carried out experiments on four representative files, which encompassed the BERT model, the ViT model \cite{dosovitskiy2020image}, CIFAR-10, and SST-2. From these files, we sample contents to generate sequence-like data sets.
    
    \textbf{Defence Model.} We introduce an LSTM-based classifier designed to discern between collisions and meaningful content. The collision sequence exhibits inconsistencies when compared to meaningful messages, which facilitates the convergence of our model. To provide a basis for comparison, we select the Naive Bayes model as our baseline model.

    \textbf{Evaluation Metrics.} In the domain and out-of-domain experiment, we focus on comparing the accuracy between different training and testing files in LSTM and Naive Bayes. In the imbalanced label experiment, we evaluate metrics including the number of samples, recall, precision, and f1-score for four types of files, both with and without JS preprocessing. 
    
    \subsubsection {Domain and Out-of-domain Experiment}  Table \ref{table:prediction} displays the classification accuracy outcomes of the detection model and baseline model trained and tested on various file types. In the diagonal, which corresponds to the domain experiments, the detection model demonstrates impeccable performance in detecting collisions when the training and testing sets are consistent, achieving a near-perfect accuracy of approximately $1.00$. Regarding the out-of-domain experiments, the collision prediction performance of the model files can be effectively transferred across different file types, resulting in an accuracy of $1.00$. However, there are certain limitations in transferring the collision prediction performance of image and text data, where the accuracy is not as high. Notably, the LSTM model exhibits superior accuracy compared to the Naive Bayes baseline.

    \begin{table}[!t]
        \centering
                \caption{Comparison of using and not using Jaccard Similarity(JS) in four types of files. We compare the number of samples, precision, recall and f1-score on $4$ different file types with and without the JS preprocessing method. }
        \label{table:JS}
        \begin{tabular}{cccccc}
            \toprule
            File Type & JS & Num of Samples & Precision & Recall & F1-score \\
            \midrule
            BERT & \ding{51} & 8279 & 1.00 & 1.00 & 1.00 \\
                 & \ding{53} & 1710961 & 1.00 & 1.00 & 1.00 \\
            ViT & \ding{51} & 2392 & 1.00 & 1.00 & 1.00 \\
                 & \ding{53} & 1357304 & 1.00 & 1.00 & 1.00 \\
            CIFAR-10 & \ding{51} & 9698 & 1.00 & 0.99 & 1.00 \\
                 & \ding{53} & 121242 & 1.00 & 0.96 & 0.98 \\
            SST-2 & \ding{51} & 5 & 1.00 & 1.00 & 1.00 \\
                 & \ding{53} & 238010 & 1.00 & 1.00 & 1.00 \\
            \bottomrule
        \end{tabular}
        \vspace{5pt}
    \end{table}

    \subsubsection{Imbalanced Label Experiment}

    In the realistic scenario, the proportion of normal contents outweighs that of collision contents, resulting in imbalanced label tasks. Additionally, there is a tremendous number of samples, which necessitates larger resources for detection. To address this issue, we have employed the Jaccard Similarity (JS) technique to filter out the abundance of negative examples. We conducted domain experiments using training and testing data from the same type of file, and the outcomes are presented in Table \ref{table:JS}. The results indicate that the classifier effectively detects collision attacks. The JS method significantly reduces the number of samples to be detected, with only a negligible impact on precision and recall.

    \section{Related Work}
    Due to space limitations, here we just briefly survey the highly related works.

    \subsection{MD5 Collision}
    Early MD5 collision research focused on identical prefix collisions (IPC) \cite{md5break,klima2005finding} by appending different suffixes to a given prefix to generate matching MD5 hashes. In 2007, chosen-prefix collisions (CPC) were introduced, revealing sophisticated abuse scenarios. Notably, CPC was utilized to create two X.509 certificates with distinct identities but the same signature from a Certification Authority \cite{collisions}. CPC has diverse applications in colliding documents, hash-based commitments, content-addressed storage, and file integrity checking \cite{cpcApplication}. Although MD5 and SHA-1 have been cryptographically broken \cite{collisions,leurent2020sha}, they are still widely used in the Deep Learning field, including popular data sets download websites such as CIFAR \cite{krizhevsky2009learning}, LFW \cite{Huang2012a}, and certain official package codes, which poses a potential threat to hash-based commitments and file integrity checking.

    Despite collision attacks being commonly employed in various applications, particularly in the software field, they are rarely utilized in the Deep Learning field, especially for PTMs. Moreover, the size of the collision is slightly larger than the original, limiting its covert nature. Therefore, we propose enhanced collision attacks to achieve equal-sized collisions.
    
    \subsection{Deep Learning Attack Methods}
    Security attacks in Deep Learning can be categorized based on the timing of their occurrence. If an attack takes place during the training phase, it is referred to as a poisoning attack, which includes backdoor attacks \cite{RIPPLe,NeuBA,shen2021backdoor} and data poisoning \cite{steinhardt2017certified,huang2020metapoison,geiping2020witches}. If it occurs during the inference phase, it is termed an evasion attack, exemplified by adversarial examples \cite{CW,PGD,black-box}. Poisoning attacks compromise the training process by corrupting the data with malicious examples, while evasion attacks employ adversarial examples to disrupt the entire classification process \cite{bae2018security}.

    Unlike adversarial examples, which only require input to launch an attack, both data poisoning and backdoor attacks depend on the attacker's ability to compromise the victim's data or model. Despite the considerable success of these attack methods, even against some large PTMs, they often overlook the aspect of stealthiness in attacks. These attacks assume that the attacker can infiltrate the victim's data sets or models, which may not always be feasible in real-world scenarios, especially within private network environments.

    Currently, individuals typically download PTMs or data sets from online public sources, using MD5 hash commitment to verify file integrity. In light of this practice, we propose the utilization of MD5 collision attacks in deep learning to enhance the covert nature of attacks.

    \subsection{Defence Against Collision Attacks}
    Several non-deep learning methods have been proposed to defend against MD5 or SHA-1 collision attacks. These methods typically involve detecting the last near-collision block of an attack and rely on key observations from the literature on MD5 and SHA-1 cryptanalysis \cite{stevens2013counter}. However, conventional defensive methods are often specific to a particular hash function, and require a deep understanding of the attack methods, which can be resource-intensive and demanding.

    To address this, we propose a general defence strategy for DL files, such as data sets and models, by leveraging deep learning methods to identify irregular contents. This approach offers a broader and more adaptive defence mechanism in the context of deep learning.
    
    \section{Conclusion}
    In this paper, we propose a novel framework for an invisible attack on AI models with enhanced MD5 collision. In contrast to conventional attacks on deep learning models, the proposed new attack is flexible, covert, and model-independent since the poisoned model has ``no difference" from the clean model, making it difficult to identify in real-world scenarios. Beyond, we have further proposed a defensive method that converts the detection task into a sequence classification problem, which has shown promising results in detecting collision attacks. We strongly recommend that MD5 be abandoned as a means of verifying the integrity of AI models. We cannot rule out that there are other more hidden and dangerous attack methods based on the vulnerability of MD5, and we also call upon the community to pay attention to this security issue and come up with strategies and regulations to reduce the risk associated with distributing and sharing AI models.
    \bibliographystyle{acm}
    \bibliography{ref}

\begin{thebibliography}{10}

\bibitem{bae2018security}
{\sc Bae, H., Jang, J., Jung, D., Jang, H., Ha, H., Lee, H., and Yoon, S.}
\newblock Security and privacy issues in deep learning.
\newblock {\em arXiv preprint arXiv:1807.11655\/} (2018).

\bibitem{bao2021beit}
{\sc Bao, H., Dong, L., and Wei, F.}
\newblock Beit: Bert pre-training of image transformers.
\newblock {\em arXiv preprint arXiv:2106.08254\/} (2021).

\bibitem{CW}
{\sc Carlini, N., and Wagner, D.}
\newblock Towards evaluating the robustness of neural networks.
\newblock In {\em 2017 ieee symposium on security and privacy (sp)\/} (2017),
  Ieee, pp.~39--57.

\bibitem{BERT}
{\sc Devlin, J., Chang, M.-W., Lee, K., and Toutanova, K.}
\newblock Bert: Pre-training of deep bidirectional transformers for language
  understanding.
\newblock {\em arXiv preprint arXiv:1810.04805\/} (2018).

\bibitem{devlin2019bert}
{\sc Devlin, J., Chang, M.-W., Lee, K., and Toutanova, K.}
\newblock {BERT}: Pre-training of deep bidirectional transformers for language
  understanding.
\newblock In {\em Proceedings of the 2019 Conference of the North American
  Chapter of the Association for Computational Linguistics: Human Language
  Technologies\/} (2019), pp.~4171--4186.

\bibitem{doan2021lira}
{\sc Doan, K., Lao, Y., Zhao, W., and Li, P.}
\newblock Lira: Learnable, imperceptible and robust backdoor attacks.
\newblock In {\em Proceedings of the IEEE/CVF International Conference on
  Computer Vision\/} (2021), pp.~11966--11976.

\bibitem{dong2019efficient}
{\sc Dong, Y., Su, H., Wu, B., Li, Z., Liu, W., Zhang, T., and Zhu, J.}
\newblock Efficient decision-based black-box adversarial attacks on face
  recognition.
\newblock In {\em Proceedings of the IEEE/CVF Conference on Computer Vision and
  Pattern Recognition\/} (2019), pp.~7714--7722.

\bibitem{dosovitskiy2020image}
{\sc Dosovitskiy, A., Beyer, L., Kolesnikov, A., Weissenborn, D., Zhai, X.,
  Unterthiner, T., Dehghani, M., Minderer, M., Heigold, G., Gelly, S., et~al.}
\newblock An image is worth 16x16 words: Transformers for image recognition at
  scale.
\newblock {\em arXiv preprint arXiv:2010.11929\/} (2020).

\bibitem{duan2021adversarial}
{\sc Duan, R., Mao, X., Qin, A.~K., Chen, Y., Ye, S., He, Y., and Yang, Y.}
\newblock Adversarial laser beam: Effective physical-world attack to dnns in a
  blink.
\newblock In {\em Proceedings of the IEEE/CVF Conference on Computer Vision and
  Pattern Recognition\/} (2021), pp.~16062--16071.

\bibitem{eykholt2018robust}
{\sc Eykholt, K., Evtimov, I., Fernandes, E., Li, B., Rahmati, A., Xiao, C.,
  Prakash, A., Kohno, T., and Song, D.}
\newblock Robust physical-world attacks on deep learning visual classification.
\newblock In {\em Proceedings of the IEEE conference on computer vision and
  pattern recognition\/} (2018), pp.~1625--1634.

\bibitem{geiping2020witches}
{\sc Geiping, J., Fowl, L., Huang, W.~R., Czaja, W., Taylor, G., Moeller, M.,
  and Goldstein, T.}
\newblock Witches' brew: Industrial scale data poisoning via gradient matching.
\newblock {\em arXiv preprint arXiv:2009.02276\/} (2020).

\bibitem{he2016deep}
{\sc He, K., Zhang, X., Ren, S., and Sun, J.}
\newblock Deep residual learning for image recognition.
\newblock In {\em Proceedings of the IEEE conference on computer vision and
  pattern recognition\/} (2016), pp.~770--778.

\bibitem{hochreiter1997long}
{\sc Hochreiter, S., and Schmidhuber, J.}
\newblock Long short-term memory.
\newblock {\em Neural computation 9}, 8 (1997), 1735--1780.

\bibitem{Huang2012a}
{\sc Huang, G.~B., Mattar, M., Lee, H., and Learned-Miller, E.}
\newblock Learning to align from scratch.
\newblock In {\em NIPS\/} (2012).

\bibitem{huang2020metapoison}
{\sc Huang, W.~R., Geiping, J., Fowl, L., Taylor, G., and Goldstein, T.}
\newblock Metapoison: Practical general-purpose clean-label data poisoning.
\newblock {\em Advances in Neural Information Processing Systems 33\/} (2020),
  12080--12091.

\bibitem{jia2022adv}
{\sc Jia, S., Yin, B., Yao, T., Ding, S., Shen, C., Yang, X., and Ma, C.}
\newblock Adv-attribute: Inconspicuous and transferable adversarial attack on
  face recognition.
\newblock {\em arXiv preprint arXiv:2210.06871\/} (2022).

\bibitem{kirillov2023segany}
{\sc Kirillov, A., Mintun, E., Ravi, N., Mao, H., Rolland, C., Gustafson, L.,
  Xiao, T., Whitehead, S., Berg, A.~C., Lo, W.-Y., Doll{\'a}r, P., and
  Girshick, R.}
\newblock Segment anything.
\newblock {\em arXiv:2304.02643\/} (2023).

\bibitem{klima2005finding}
{\sc Klima, V.}
\newblock Finding md5 collisions on a notebook pc using multi-message
  modifications.
\newblock {\em Cryptology ePrint Archive\/} (2005).

\bibitem{krizhevsky2009learning}
{\sc Krizhevsky, A., Hinton, G., et~al.}
\newblock Learning multiple layers of features from tiny images.

\bibitem{RIPPLe}
{\sc Kurita, K., Michel, P., and Neubig, G.}
\newblock Weight poisoning attacks on pre-trained models.
\newblock {\em arXiv preprint arXiv:2004.06660\/} (2020).

\bibitem{leurent2019collisions}
{\sc Leurent, G., and Peyrin, T.}
\newblock From collisions to chosen-prefix collisions application to full
  sha-1.
\newblock In {\em Annual International Conference on the Theory and
  Applications of Cryptographic Techniques\/} (2019), Springer, pp.~527--555.

\bibitem{leurent2020sha}
{\sc Leurent, G., and Peyrin, T.}
\newblock Sha-1 is a shambles: First chosen-prefix collision on sha-1 and
  application to the pgp web of trust.
\newblock In {\em Proceedings of the 29th USENIX Conference on Security
  Symposium\/} (2020), pp.~1839--1856.

\bibitem{liu2020roberta}
{\sc Liu, Y., Ott, M., Goyal, N., Du, J., Joshi, M., Chen, D., Levy, O., Lewis,
  M., Zettlemoyer, L., and Stoyanov, V.}
\newblock Ro{BERT}a: A robustly optimized {BERT} pretraining approach.
\newblock In {\em arXiv preprint arXiv:1907.11692\/} (2020).

\bibitem{PGD}
{\sc Madry, A., Makelov, A., Schmidt, L., Tsipras, D., and Vladu, A.}
\newblock Towards deep learning models resistant to adversarial attacks.
\newblock {\em arXiv preprint arXiv:1706.06083\/} (2017).

\bibitem{ouyang2022training}
{\sc Ouyang, L., Wu, J., Jiang, X., Almeida, D., Wainwright, C.~L., Mishkin,
  P., Zhang, C., Agarwal, S., Slama, K., Ray, A., et~al.}
\newblock Training language models to follow instructions with human feedback.
\newblock {\em arXiv preprint arXiv:2203.02155\/} (2022).

\bibitem{black-box}
{\sc Papernot, N., McDaniel, P., Goodfellow, I., Jha, S., Celik, Z.~B., and
  Swami, A.}
\newblock Practical black-box attacks against machine learning.
\newblock In {\em Proceedings of the 2017 ACM on Asia conference on computer
  and communications security\/} (2017), pp.~506--519.

\bibitem{radford2019language}
{\sc Radford, A., Wu, J., Child, R., Luan, D., Amodei, D., and Sutskever, I.}
\newblock Language models are unsupervised multitask learners.
\newblock {\em OpenAI\/} (2019).

\bibitem{saha2020hidden}
{\sc Saha, A., Subramanya, A., and Pirsiavash, H.}
\newblock Hidden trigger backdoor attacks.
\newblock In {\em Proceedings of the AAAI conference on artificial
  intelligence\/} (2020), vol.~34, pp.~11957--11965.

\bibitem{saharia2022photorealistic}
{\sc Saharia, C., Chan, W., Saxena, S., Li, L., Whang, J., Denton, E.,
  Ghasemipour, S. K.~S., Ayan, B.~K., Mahdavi, S.~S., Lopes, R.~G., et~al.}
\newblock Photorealistic text-to-image diffusion models with deep language
  understanding.
\newblock {\em arXiv preprint arXiv:2205.11487\/} (2022).

\bibitem{shafahi2018poison}
{\sc Shafahi, A., Huang, W.~R., Najibi, M., Suciu, O., Studer, C., Dumitras,
  T., and Goldstein, T.}
\newblock Poison frogs! targeted clean-label poisoning attacks on neural
  networks.
\newblock {\em Advances in neural information processing systems 31\/} (2018).

\bibitem{sharma2019cryptography}
{\sc Sharma, A.~K., and Mittal, S.}
\newblock Cryptography \& network security hash function applications, attacks
  and advances: A review.
\newblock In {\em 2019 Third International Conference on Inventive Systems and
  Control (ICISC)\/} (2019), IEEE, pp.~177--188.

\bibitem{shen2021backdoor}
{\sc Shen, L., Ji, S., Zhang, X., Li, J., Chen, J., Shi, J., Fang, C., Yin, J.,
  and Wang, T.}
\newblock Backdoor pre-trained models can transfer to all.
\newblock {\em arXiv preprint arXiv:2111.00197\/} (2021).

\bibitem{socher2013recursive}
{\sc Socher, R., Perelygin, A., Wu, J., Chuang, J., Manning, C.~D., Ng, A.~Y.,
  and Potts, C.}
\newblock Recursive deep models for semantic compositionality over a sentiment
  treebank.
\newblock In {\em Proceedings of the 2013 conference on empirical methods in
  natural language processing\/} (2013), pp.~1631--1642.

\bibitem{standardnational}
{\sc Standard, S.~H.}
\newblock National institute of standards and technology (nist), fips
  publication 180-2, aug 2002, 2002.

\bibitem{steinhardt2017certified}
{\sc Steinhardt, J., Koh, P. W.~W., and Liang, P.~S.}
\newblock Certified defenses for data poisoning attacks.
\newblock {\em Advances in neural information processing systems 30\/} (2017).

\bibitem{stevens2013counter}
{\sc Stevens, M.}
\newblock Counter-cryptanalysis.
\newblock In {\em Advances in Cryptology--CRYPTO 2013: 33rd Annual Cryptology
  Conference, Santa Barbara, CA, USA, August 18-22, 2013. Proceedings, Part
  I\/} (2013), Springer, pp.~129--146.

\bibitem{stevens2017first}
{\sc Stevens, M., Bursztein, E., Karpman, P., Albertini, A., and Markov, Y.}
\newblock The first collision for full sha-1.
\newblock In {\em Annual international cryptology conference\/} (2017),
  Springer, pp.~570--596.

\bibitem{collisions}
{\sc Stevens, M., et~al.}
\newblock On collisions for md5, 2007.

\bibitem{cpcApplication}
{\sc Stevens, M., Lenstra, A.~K., and De~Weger, B.}
\newblock Chosen-prefix collisions for md5 and applications.
\newblock {\em International Journal of Applied Cryptography 2}, ARTICLE
  (2012), 322--359.

\bibitem{md5break}
{\sc Wang, X., and Yu, H.}
\newblock How to break md5 and other hash functions.
\newblock In {\em Annual international conference on the theory and
  applications of cryptographic techniques\/} (2005), Springer, pp.~19--35.

\bibitem{yang2017research}
{\sc Yang, Y., Zhang, X., Yu, J., Zhang, P., et~al.}
\newblock Research on the hash function structures and its application.
\newblock {\em Wireless Personal Communications 94}, 4 (2017), 2969--2985.

\bibitem{NeuBA}
{\sc Zhang, Z., Xiao, G., Li, Y., Lv, T., Qi, F., Liu, Z., Wang, Y., Jiang, X.,
  and Sun, M.}
\newblock Red alarm for pre-trained models: Universal vulnerability to
  neuron-level backdoor attacks.
\newblock {\em arXiv preprint arXiv:2101.06969\/} (2021).

\bibitem{zhou2021ibot}
{\sc Zhou, J., Wei, C., Wang, H., Shen, W., Xie, C., Yuille, A., and Kong, T.}
\newblock ibot: Image bert pre-training with online tokenizer.
\newblock {\em arXiv preprint arXiv:2111.07832\/} (2021).

\end{thebibliography}
    
\end{document}